\documentclass[sigconf, nonacm]{acmart}

\usepackage{enumitem}
\usepackage{url}
\usepackage{tabularx}
\usepackage{tablefootnote}

\newcommand\vldbdoi{XX.XX/XXX.XX}
\newcommand\vldbpages{XXX-XXX}
\newcommand\vldbvolume{14}
\newcommand\vldbissue{1}
\newcommand\vldbyear{2020}
\newcommand\vldbauthors{\authors}
\newcommand\vldbtitle{\shorttitle} 

\newcommand\vldbpagestyle{plain} 

\begin{document}
\title{LIDER: An Efficient High-dimensional Learned Index for Large-scale Dense Passage Retrieval}

\author{Yifan Wang}
\affiliation{%
  \institution{University of Florida}
}
\email{wangyifan@ufl.edu}

\author{Haodi Ma}
\affiliation{%
  \institution{University of Florida}
}
\email{ma.haodi@ufl.edu}

\author{Daisy Zhe Wang}
\affiliation{%
  \institution{University of Florida}
}
\email{daisyw@ufl.edu}

\begin{abstract}
Passage retrieval has been studied for decades, and many recent approaches of passage retrieval are using dense embeddings generated from deep neural models, called ``dense passage retrieval''. The state-of-the-art end-to-end dense passage retrieval systems normally deploy a deep neural model followed by an approximate nearest neighbor (ANN) search module. The model generates embeddings of the corpus and queries, which are then indexed and searched by the high-performance ANN module. 
With the increasing data scale, the ANN module unavoidably becomes the bottleneck on efficiency.
An alternative is the learned index, which achieves significantly high search efficiency by learning the data distribution and predicting the target data location.   
But most of the existing learned indexes are designed for low dimensional data, which are not suitable for dense passage retrieval with high-dimensional dense embeddings.

In this paper, we propose \textbf{LIDER}, an efficient high-dimensional \textbf{L}earned \textbf{I}ndex for large-scale \textbf{DE}nse passage \textbf{R}etrieval. LIDER has a clustering-based hierarchical architecture formed by two layers of core models. As the basic unit of LIDER to index and search data, a \textit{core model} includes an adapted recursive model index (RMI) and a dimension reduction component which consists of an extended SortingKeys-LSH (SK-LSH) and a key re-scaling module. The dimension reduction component reduces the high-dimensional dense embeddings into one-dimensional keys and sorts them in a specific order, which are then used by the RMI to make fast prediction.
Experiments show that LIDER has a higher search speed with high retrieval quality comparing to the state-of-the-art ANN indexes on passage retrieval tasks, e.g., on large-scale data it achieves 1.2x search speed and significantly higher retrieval quality than the fastest baseline in our evaluation. Furthermore, LIDER has a better capability of speed-quality trade-off.

\end{abstract}

\maketitle

\pagestyle{\vldbpagestyle}
\begingroup\small\noindent\raggedright\textbf{PVLDB Reference Format:}\\
\vldbauthors. \vldbtitle. PVLDB, \vldbvolume(\vldbissue): \vldbpages, \vldbyear.\\
\href{https://doi.org/\vldbdoi}{doi:\vldbdoi}
\endgroup
\begingroup
\renewcommand\thefootnote{}\footnote{\noindent
This work is licensed under the Creative Commons BY-NC-ND 4.0 International License. Visit \url{https://creativecommons.org/licenses/by-nc-nd/4.0/} to view a copy of this license. For any use beyond those covered by this license, obtain permission by emailing \href{mailto:info@vldb.org}{info@vldb.org}. Copyright is held by the owner/author(s). Publication rights licensed to the VLDB Endowment. \\
\raggedright Proceedings of the VLDB Endowment, Vol. \vldbvolume, No. \vldbissue\ %
ISSN 2150-8097. \\
\href{https://doi.org/\vldbdoi}{doi:\vldbdoi} \\
}\addtocounter{footnote}{-1}\endgroup

\section{Introduction}
\label{sec:introduction}
As one of the most important types of information retrieval, passage retrieval finds and returns relevant passages to a given query. The typical applications of passage retrieval include question answering (QA) systems~\cite{DPR-karpukhin-etal-2020} (which retrieve relevant answers to the questions), dialogue response selection~\cite{dense-retrieval-app-dialog-response-sel} (which selects proper responses given a dialogue context and multiple response candidates) and search engines.
A passage retrieval pipeline normally consists of two stages, the first-stage retrieval and second-stage reranking, where the former retrieves a collection of candidate passages and the latter reranks them by relevance to the query. In this paper we focus on the first-stage retrieval and simply denote it by \textit{retrieval}. A typical implementation of the retrieval uses bag-of-word vectors to represent the query and passages, where each element in a vector stands for a term in the vocabulary and length of the vector equals the vocabulary size. So the representation is usually very sparse and such an implementation is called \textit{sparse retrieval}. In recent years \textit{dense retrieval} has emerged and shown its great potential in effective passage retrieval. This new implementation of passage retrieval uses dense neural embeddings to represent the text, by which significant improvement has been made, due to the capability of neural embeddings on capturing semantic information. One of the biggest problems in sparse retrieval is the \textit{vocabulary mismatch}. When the query and passage use different terms (like synonyms) to express similar meanings, term-based sparse retrieval cannot match them. But in such a case the neural embeddings of the query and passage are likely to be still close to each other. Therefore more and more passage retrieval systems are deploying dense retrieval today.
Specifically, the state-of-the-art dense retrieval pipelines commonly consist of a deep neural model (which is normally a two-tower deep model, e.g., two-tower BERT) and an approximate nearest neighbor (ANN) index~\cite{DPR-karpukhin-etal-2020, distilling-dense-embed-lin2020, colbert}. The embeddings of the corpus are generated by the model and indexed by the ANN module offline, then online the embeddings of queries are computed and their top-k nearest neighbors are retrieved by the efficient ANN index from the corpus embeddings as results, by which the pipelines can support low-latency online serving.

But no matter how high-performance the ANN index is on specific sizes of datasets, with the continuous data explosion today, it still becomes the bottleneck of retrieval speed. Given a query, the retrieval time mainly consists of two parts, the query embedding generation time and ANN search time, where the former takes $0.0007 \sim 0.0008$ seconds while the latter takes 0.04 $\sim$ more than 0.2 seconds in our evaluation, i.e., the ANN search usually costs more than 50x time of the query embedding generation.

A potential high-performance alternative to the typical ANN indexes~\cite{ivfadc, pq-inverted-multi-index-2012, Graph-based-index-NN-Descent-2011, Graph-based-index-GNNS-hajebi2011} is the learned index.
Unlike traditional index techniques, instead of ``looking up'' the location of a key, the learned index ``predicts'' the location after learning the key-location distribution of the dataset, which makes its search process highly efficient. As the first work to propose the concept ``learned index structure'', \cite{case-for-learned-index-structures-kraska2018case} introduces a learned index structure called \textit{Recursive-Model Index} (RMI) that outperforms traditional B-tree index on efficiency and memory usage for range queries on one-dimensional data. Specifically, RMI is a hierarchical structure including multiple layers of machine learning models. Given data and their corresponding indexing keys, RMI assumes the data records are sorted as a dense array by the keys and learns the location distribution of the keys in the array. Then for any given key, RMI can predict its location in the array with a bounded error. In the case that prediction error is beyond the required bound, a hybrid index mixing RMI and B-tree will be built to reduce the overall error.

Since RMI requires the data records to be stored in a one-dimensional sorted array, it cannot be directly applied to multi-dimensional data as such data has no natural order for sorting. Therefore, following studies normally utilize dimension reduction methods to convert the multi-dimensional data into sortable one-dimensional points for RMI to fit, based on which 
several multi-dimensional learned indexes are proposed, e.g., ZM index~\cite{zm-index}, ML-Index~\cite{ML-index-davitkova2020ml}, Flood~\cite{flood-index}, etc.  But most of them are still designed for very low-dimensional data, like 2D and 3D spatial data or multi-dimensional data with only tens of dimensions. On dense neural embeddings with hundreds to thousands of dimensions, their dimension reduction methods are no longer effective due to the curse of dimensionality, making them not suitable for dense retrieval tasks.

An effective dimension reduction method for high-dimensional data is locality-sensitive hashing (LSH). It is able to convert high-dimensional data into one-dimensional hash code string, which is called ``hashkey'' in this paper.
But there is still a gap between LSH and learned index: RMI not only requires one-dimensional data, but also requires the data can be sorted meaningfully. Unfortunately, there is no explicit order for the hashkeys generated by LSH. This problem can be solved by SortingKeys-LSH (SK-LSH)~\cite{sklsh-liu2014sk}, which defines a specialized order on the LSH hashkeys such that the data points can be sorted by their hashkeys to form a continuous array where the positions meaningfully reflect the similarities between the data points.

However, to make the index system practical, we must solve one more problem which we call ``the curse of space size'': when search space (i.e., dataset size) becomes larger, both of SK-LSH and RMI perform worse. Specifically, SK-LSH needs a larger hashkey length and more sorted arrays to guarantee the effectiveness, leading to a significantly larger memory usage, and the capability of RMI to accurately fit a dataset normally degrades with the increasing dataset size. Our solution to this problem is the clustering. By clustering the whole dataset into smaller groups and building index inside each group, the search space of each index is shrunk effectively. Furthermore, this solution can benefit from parallelization as the clusters are independent from each other.

In this paper we propose \textbf{LIDER}, a novel high-dimensional \textbf{L}earned \textbf{I}ndex for efficient large-scale \textbf{DE}nse passage \textbf{R}etrieval. 
LIDER has a clustering-based two-layer architecture that consists of multiple core models. As the basic unit in LIDER, a \textit{core model} is a module combining modified RMI and adapted SK-LSH together. Specifically, a core model mainly consists of two components, a couple of simplified RMIs and an extended SK-LSH module (called ESK-LSH). The ESK-LSH module works on reducing data dimension, i.e., converting the high-dimensional dense embeddings to one-dimensional hashkeys and sorting them to form several sorted arrays by a specialized order, while the RMIs learn the hashkey-location distributions in those sorted arrays. During ANN search, in a core model, ESK-LSH converts the query embedding into hashkeys and feeds them to RMIs to predict the locations, starting from where a range search on the sorted arrays will finally retrieve the approximate nearest neighbors. 
To have the two major components better interact with each other and better fit the demands of dense retrieval, we make several critical adaptions and improvements.
The technical details are included in Section~\ref{sec:LIDER}, \ref{sec:esk-lsh} and \ref{sec:rmi}.

The two-layer structure of LIDER is built as follows: first the target dataset is clustered, then one core model is created over the cluster centroids, called ``centroids retriever'' and forming the first layer. And one core model per cluster is created to index the data inside that cluster, called ``in-cluster retriever''. All these in-cluster retrievers fill the second layer. During the retrieval process, the centroids retriever is first called to find the target clusters (i.e., the clusters possibly including correct answers to the query) and then the in-cluster retrievers in those targets will retrieve the final results. In such a way, each single core model works on a smaller subspace, which effectively tackles ``the curse of space size''. And the retrieval in this architecture is simple to be parallelized between clusters since the in-cluster retrievers work independently from each other.

We conduct evaluation based on the passage retrieval scenario, which shows that LIDER is highly effective with higher efficiency than the state-of-the-art ANN indexes on large-scale dense retrieval. Comparing to the fastest baseline method, LIDER achieves significantly higher retrieval quality with 1.2x search speed on the largest evaluation dataset.
To our best knowledge, LIDER is one of the first learned index structures for ANN queries on high-dimensional data. And LIDER is the first implementation of such a learned index to solve dense retrieval tasks.

The main contributions of this paper are shown below:
\begin{enumerate}[leftmargin=2em]
\item We build LIDER, an efficient and effective high-dimensional learned index for ANN search in dense retrieval. This is one of the first learned indexes for ANN search on very high-dimensional data, and one of the first dense retrieval solutions utilizing learned index.
\item We design a clustering-based two-layer hierarchical architecture to address the performance issues raised by large search space and optimize the retrieval efficiency by parallelization between clusters. 
\item We extend SK-LSH to more distance metrics, improve its parallelism and hashkey sorting method, design an effective re-scaling method on the hashkeys to better train RMI, and simplify RMI for better predicting efficiency and effectiveness.
\item We conduct experiments based on commonly used passage retrieval benchmarks to evaluate the performance of LIDER, 
which shows LIDER outperforms all baseline methods on both efficiency and effectiveness, especially for large-scale data.
\end{enumerate}

This paper is organized as follows: Section~\ref{sec:related-work} introduces works related to LIDER. Section~\ref{sec:LIDER} introduces the architecture and workflow of LIDER. The following Sections \ref{sec:esk-lsh} and \ref{sec:rmi} present technical details for the major components. Section \ref{sec:time-complexity} provides the time complexity analysis. Finally Section~\ref{sec:exp} presents experiments to evaluate the performance of LIDER, the effect of some important factors and parameters, and the memory footprint and index construction cost.   

\section{Related work}
\label{sec:related-work}
There are two major categories of retrieval methods, sparse and dense retrieval. Sparse retrieval normally uses bag-of-word (BOW) models, where the document representations are sparse vectors, while dense retrieval mostly utilizes neural embeddings from deep neural models which are dense vectors. Typical sparse retrieval methods include BM25, DeepCT~\cite{deepct}, Doc2query~\cite{doc2query} and docTTTTTquery~\cite{docTTTTTquery}, which are commonly used in recent retrieval studies as strong baselines. Due to the power of dense neural embeddings in semantic search, many state-of-the-art retrieval researches focus on dense retrieval. After BERT~\cite{bert} was proposed, most of the recent dense retrieval models are designed based on it and achieve significant improvement on retrieval quality, e.g., Sentence-BERT~\cite{sentence-bert} and MarkedBERT~\cite{markedbert}. But as a heavy model, BERT has high inference latency, which limits its application on online retrieval that requires low-latency serving. To solve this problem, following works have proposed several variants to reduce its complexity, including DistilBERT~\cite{distilbert}, ColBERT~\cite{colbert}, TCT-ColBERT~\cite{tct-colbert}, etc.   

To better support the low-latency online retrieval, in addition to deploying more lightweight neural models, most state-of-the-art end-to-end dense retrieval systems also arrange a high-performance ANN search module following the neural model to fast look up the closest documents to the queries based on their embeddings. ANN indexes include four major categories, i.e., hashing, graph, quantization and tree based indexes. Among them the tree based indexes are more suitable to low-dimensional space, so dense retrieval systems normally choose from the other three types of indexes. For example, DPR~\cite{DPR-karpukhin-etal-2020} utilizes a graph based index, HNSW~\cite{Graph-based-index-hnsw-malkov2016}, ColBERT deploys IVFADC index which is based on product quantization,  BPR~\cite{bpr} integrates learning-to-hash technique into DPR~\cite{DPR-karpukhin-etal-2020}, etc. FAISS~\cite{faiss} is one of the most popular ANN index libraries in today's dense retrieval, as it implements high-performance indexes of all the three classes. It has been used by many recent dense retrieval studies~\cite{dense-retrieval-with-faiss-tang2021improving, dense-retrieval-with-faiss-wu2019scalable, dense-retrieval-with-faiss-zhan2021learning, dense-retrieval-with-faiss-ANN, dense-retrieval-with-faiss-ANN-2}. There are also many other representative ANN indexes \cite{ivfadc, pq-inverted-multi-index-2012, Graph-based-index-NN-Descent-2011, Graph-based-index-GNNS-hajebi2011}.

Learned index is first proposed by \cite{case-for-learned-index-structures-kraska2018case}, where the structure is called \textit{Recursive-model index} (RMI). RMI is designed to replace the traditional range search indexes on one-dimensional data. To handle multi-dimensional data, following works propose many multi-dimensional learned index structures based on RMI. Their main differences are on the dimension reduction methods used to reduce the multi-dimensional data into one-dimension.  SageDB~\cite{sagedb-kraska2019} uses a dimension reduction method with some similarities to LSH (which is not explicitly discussed in the paper). ZM-index~\cite{zm-index}, ML-index~\cite{ML-index-davitkova2020ml} and Flood~\cite{flood-index} reduce data dimension by Z-order, iDistance, and a multi-dimensional grid based method respectively. \cite{case-for-learned-spatial-indexes-pandey2020} applies the idea of Flood to several traditional multi-dimensional indexes to make them better handle spatial queries. \cite{Learned-Index-for-Exact-Similarity-Search-in-Metric-Spaces} proposes LIMS, a learned index for efficient similarity search in metric spaces. LIMS applies a clustering-based strategy to split the dataset into more uniformly distributed subsets to improve the search performance. Though LIDER also has a clustering-based architecture, it is quite different from LIMS: (1) LIMS is an exact similarity search index while LIDER is for approximate similarity search, (2) LIMS is a disk-based index while LIDER is in-memory index, and (3) the data dimension of LIMS is not very high, i.e., less than 100 in their experiments, which is significantly lower than that of LIDER. 
To our best knowledge, all these existing learned indexes mentioned above are designed for relatively low-dimensional data and there is no learned index for very high-dimensional data like neural embeddings. 

\section{LIDER}
\label{sec:LIDER}
In this section we introduce the overall architecture and workflow of LIDER and each core model in it.
Then in the following sections we discuss important technical details about the major components in LIDER. 

\begin{figure}[!h]
  \centering
  \includegraphics[width=\columnwidth]{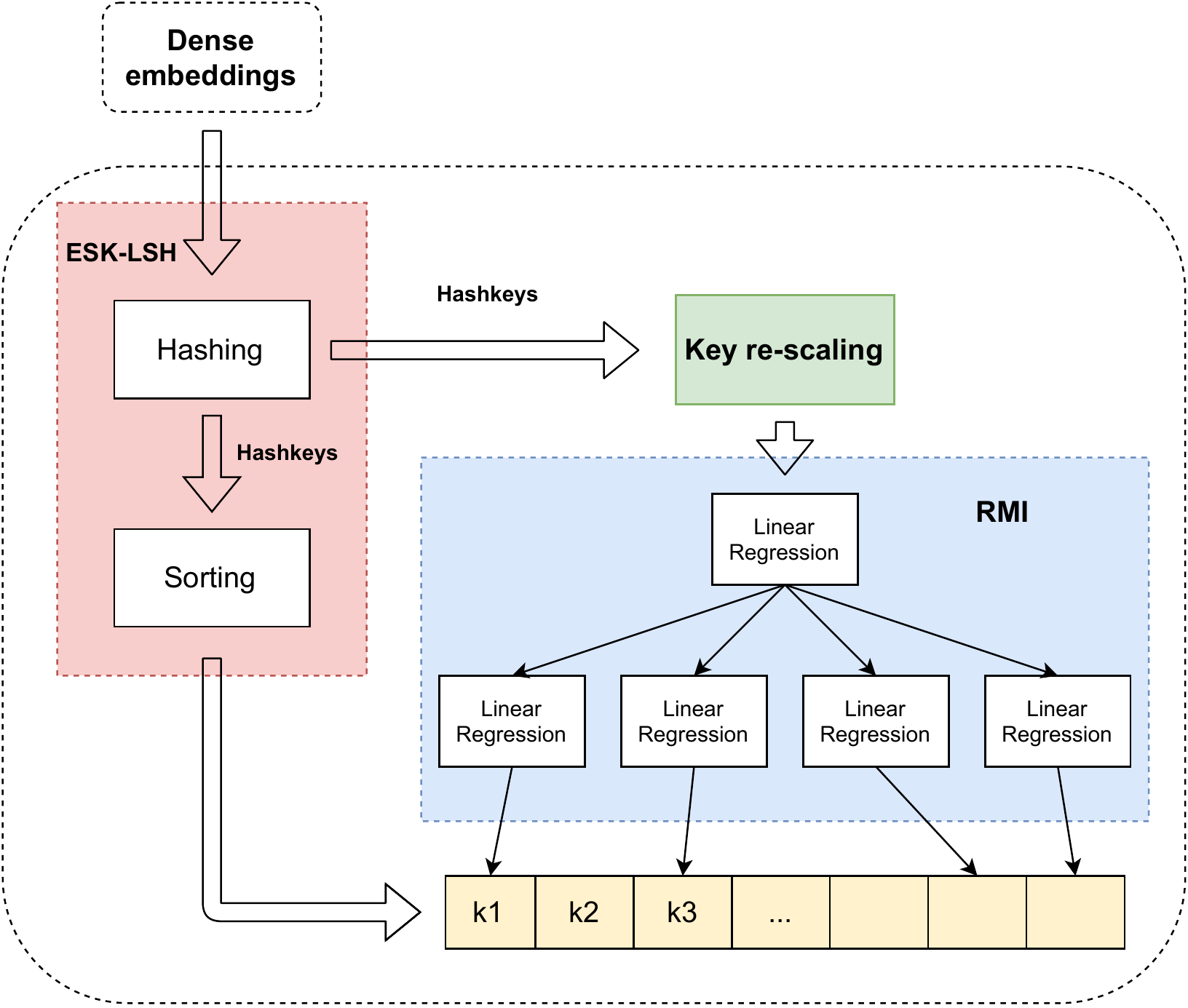}
  \caption{Core model structure}
  \setlength{\belowcaptionskip}{-50pt}
  \label{fig:arch-core-model}
\end{figure}

\begin{figure}[!h]
  \centering
  \includegraphics[width=\columnwidth]{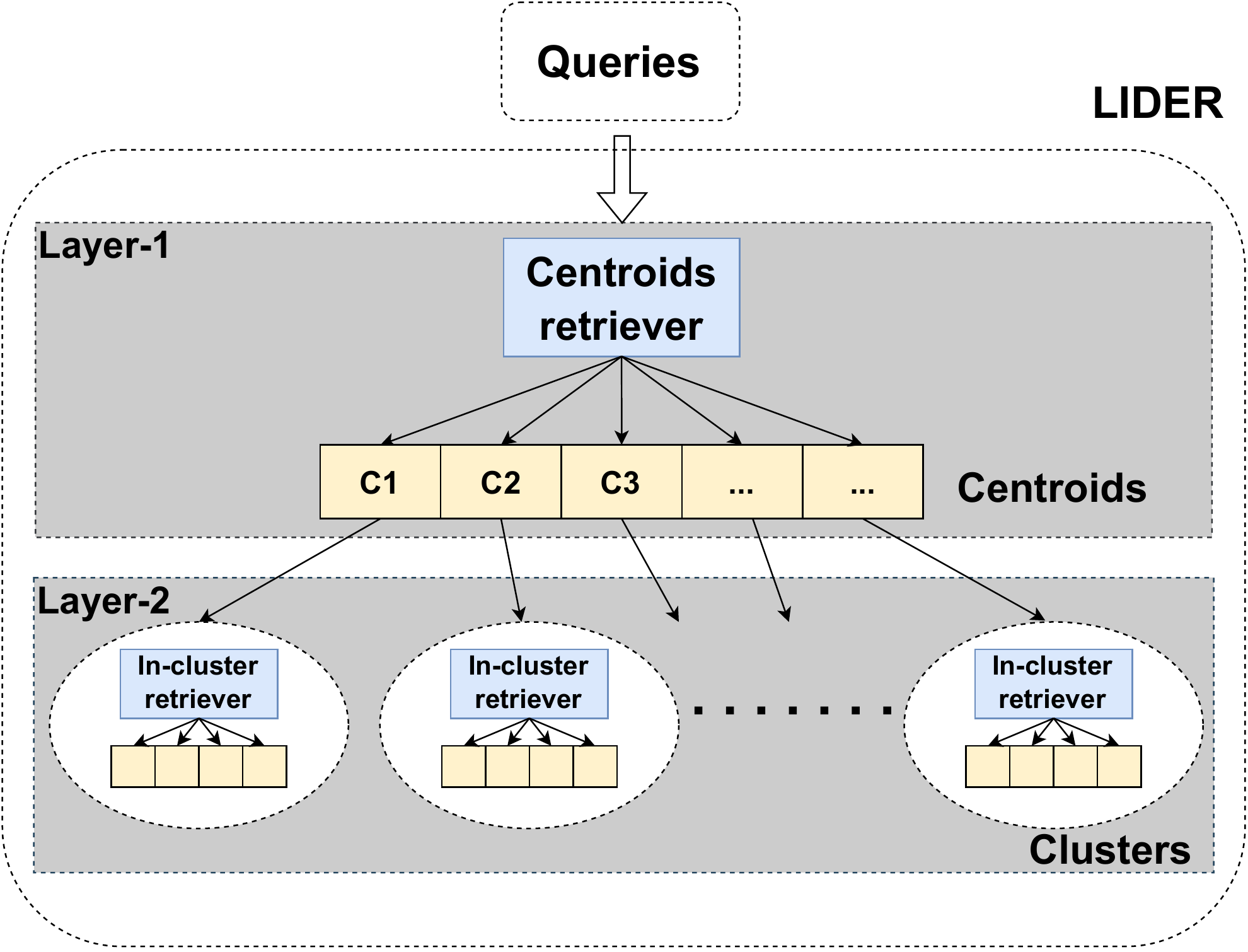}
  \caption{LIDER architecture (the \textit{centroids retriever} and each \textit{in-cluster retriever} are both a core model as shown in Figure~\ref{fig:arch-core-model})}
  \setlength{\belowcaptionskip}{-50pt}
  \label{fig:two-layer-arch}
\end{figure}

\subsection{Core model structure}
The structure of core model is illustrated in Figure~\ref{fig:arch-core-model}. There are two major modules in the core model, the dimension reduction module (illustrated as the pink and green boxes) and the location prediction module (i.e., the RMI, marked as the blue box). The dimension reduction module includes two components, the extended SK-LSH (ESK-LSH), which is extended from the original SK-LSH, and the Key re-scaling component, which is used to convert the string hashkeys into numeric keys within a proper range to train the RMIs. The yellow boxes stand for one of the sorted arrays generated by ESK-LSH. Essentially, one RMI corresponds to one sorted array. And ESK-LSH may maintain multiple sorted arrays, in which case there will be multiple RMIs existing in a core model. More details are presented in Section~\ref{sec:esk-lsh} and \ref{sec:rmi}. 

\subsection{Clustering-based architecture}
As shown in Figure~\ref{fig:two-layer-arch}, LIDER has a clustering-based two-layer architecture. The first layer (the higher grey box, marked as ``Layer-1'') includes one core model, named \textit{centroids retriever}, to index all the cluster centroids. And the second layer (the lower grey box, marked as ``Layer-2'') consists of multiple core models, one inside each data cluster, named \textit{in-cluster retriever}. Each in-cluster retriever indexes the data within the corresponding cluster.

Comparing to building only one core model to index the whole dataset, such a clustering-based layered architecture allows each core model to fit a significantly smaller search space, e.g., roughly less than 1/1000 of the whole dataset since we normally set the number of clusters to be more than 1000. A smaller dataset is better for RMI to fit and predict more accurately, and allows ESK-LSH to reduce the memory usage and search latency by shrinking the hashkey length and number of hashkey arrays while guaranteeing the effectiveness.  
In our implementation, the clusters are generated simply by k-means clustering. So unlike the components in core model, this paper does not have a specific section to show more details about the clustering strategy. All important details for it have been included in this section and Section~\ref{sec:workflow}.

\subsection{Workflow}
\label{sec:workflow}
In this section, we introduce the workflow of building LIDER and query processing in LIDER, from perspectives of single core model and the whole system respectively. 
\subsubsection{Workflow in each core model} 
For a single core model, to make it simple, here we introduce the indexing and querying workflow in the case that only one sorted array and one RMI are maintained, as it is in Figure~\ref{fig:arch-core-model}. When a core model is indexing the corpus of documents, document embeddings generated by the upstream deep neural model are first input to ESK-LSH to generate their hashkeys. Then the hashkeys are sorted to form the sorted hashkey array, and also passed to the key re-scaling component to be converted into numeric keys. For each of these numeric keys, its location is that of the original hashkey in the sorted array. These key-location pairs are then used to train the RMI, where the key is data and location is label. RMI will learn to predict the location of a given key by the training. When querying $k$ nearest neighbors with the core model, the query embedding is first input to ESK-LSH to generate the query hashkey, then query hashkey is converted into numeric query key, which will be passed to RMI to predict its location. After the prediction, ESK-LSH will do a bi-directional range search (with a pre-determined range width purely depending on $k$) starting from the predicted location on the sorted array to retrieve a pre-determined number of candidate hashkeys (where the number is normally several times of $k$). Finally the original document embeddings corresponding to the candidate hashkeys will be scored (e.g., computing cosine similarity to the query embedding) and the top-$k$ of them will be returned. In the case that multiple arrays and RMIs are used, the core model will simply do a parallel execution of such a process on each array and RMI then merge the results.   
\subsubsection{Overall workflow in LIDER}
LIDER is built as follows: first the whole dataset is clustered into several subsets by k-means clustering algorithm, then a core model is created on the collection of the cluster centroids, and within each data cluster, one core model is also created to index that subset. The core model indexing the centroids is the centroids retriever while each core model inside a cluster is an in-cluster retriever. When using LIDER for ANN queries, the query embedding is first passed into the centroids retriever to find the approximate nearest centroids to it, then the query will be searched by the in-cluster retrievers in the clusters corresponding to those centroids. This in-cluster retrieval process is parallelized between clusters, since the in-cluster retrievers work independently from each other, therefore LIDER makes them execute the retrieval simultaneously to improve the efficiency. Finally the results from in-cluster retrievers will be merged and the top-$k$ of them by some score (like cosine similarity to the query embedding) will be returned.

\section{Extended SK-LSH}
\label{sec:esk-lsh}
Locality-sensitive hashing (LSH) is an effective hashing framework for approximate similarity search on high-dimensional data. It generates a signature (called ``hashkey'') for each data vector by hashing it with several hashing functions $h_1(\cdot),...,h_m(\cdot)$ and concatenating the resulting hashing values. Then the similarity between two data vectors can be approximated using their hashkeys. 
The specialized hashing function $h_i(\cdot)$ is so called \textit{LSH function} while the sequence of them ${G(\cdot) = (h_1(\cdot),...,h_m(\cdot))}$ is normally called \textit{compound LSH function}. LSH function is specially designed to achieve both of randomization and locality-preserving. Its formal definition is:

\begin{definition}[LSH function]
Given a metric space and its distance metric $d$, a threshold $t > 0$, and any two data points $\vec{u}$ and $\vec{v}$ in the metric space, an LSH function $h(\cdot)$ should satisfy the following two conditions:

(1) if $d(\vec{u}, \vec{v}) \leq t$, then $h(\vec{u}) = h(\vec{v})$ with a probability at least $p_1$

(2) if $d(\vec{u}, \vec{v}) \geq ct$, then $h(\vec{u}) = h(\vec{v})$ with probability at most $p_2$

where $c > 1$ is an approximation factor, and the probabilities $p_1 > p_2$.

\end{definition}
The definition means an LSH function maps similar/close data points into the same hash bucket with a higher probability than mapping dissimilar points into the same bucket. As a result, the data points within the same bucket are very likely to be similar to each other. By using the compound LSH function, the false negative rate can be further reduced to increase the recall.

Based on the essential LSH model, 
SK-LSH~\cite{sklsh-liu2014sk} defines a linear order over the hashkeys to sort them such that the data points with smaller Euclidean distance are placed closer on the sorted hashkey array. 
To find nearest neighbors of the query vector, it is sufficient to find data vectors whose hashkeys are close to the query hashkey on the array, and it is easy for SK-LSH to expand the nearest neighbor search by scanning farther locations from the initial hashkey (i.e., the closest hashkey to that of the query) towards both of its left and right sides along the array, called ``bi-directional expansion'', which is basically a fixed length range search on the array, where the range length is pre-determined by $k$ (normally several times of $k$). To let SK-LSH better fit dense retrieval scenario and have a higher performance, we make several extensions and improvement on it.

\subsection{Extension on similarity metrics}
Many of the state-of-the-art embedding based research and applications~\cite{mikolov2013efficient,lacroix2018canonical,shi2019probabilistic,bert} rely on cosine similarity to accurately measure the embedding similarities. For dense retrieval, in addition to cosine similarity~\cite{DPR-karpukhin-etal-2020, colbert, curse-for-large-index-reimers2020curse}, inner product is another commonly used similarity metric~\cite{distilling-dense-embed-lin2020, ance-xiong2020approximate, dense-retrieval-with-faiss-tang2021improving}.
However, SK-LSH is designed for Euclidean distance. Therefore, we extend SK-LSH to also work with cosine similarity and name the extended algorithm as ESK-LSH. We do not make a specific extension for inner product since it is equivalent to cosine similarity on normalized vectors.

SK-LSH requires a base LSH model to perform the hashing. To support cosine similarity, we use the hyperplane-based random projection LSH~\cite{random-projection-lsh} model. Since~\cite{sklsh-liu2014sk} only guarantees the properties of SK-LSH for Euclidean distance and the corresponding LSH model, it is necessary to show those properties still hold for cosine similarity and hyperplane-based random projection LSH model.

For any two vectors $\vec{p_1}$, $\vec{p_2} \in R^d$, we define $sim(\vec{p_1}, \vec{p_2})$ as their cosine similarity:
\begin{equation}
\label{eq:cos-sim}
sim(\vec{p_1}, \vec{p_2}) := \cos(\theta_{1,2}) = \frac{\vec{p_1} \cdot \vec{p_2}}{\left\lVert\vec{p_1}\right\rVert \cdot \left\lVert\vec{p_2}\right\rVert}
\end{equation}

We also know from~\cite{random-projection-lsh} that, the probability of generating identical hash values for two vectors by the LSH model is 
\begin{equation}
\label{eq:lsh-prob-equal}
p(\theta_{1,2}) := P[h(\vec{p_1}) = h(\vec{p_2})] = 1 - \frac{\theta_{1,2}}{\pi}
\end{equation}

Following~\cite{sklsh-liu2014sk}, we use ${G(\cdot) = (h_1(\cdot),...,h_m(\cdot))}$ to denote a compound LSH function, where $h_i: R^d \rightarrow \{0, 1\}$ are randomly selected hash functions defined in~\cite{random-projection-lsh}. $\operatorname{KL}(K_1, K_2)$ is the \textit{non-prefix length}~\cite{sklsh-liu2014sk} between $K_1 = G(\vec{p_1})$ and $K_2 = G(\vec{p_2})$. $\operatorname{dist}(K_1, K_2)$ is the distance between $K_1$ and $K_2$. Please refer to Equation~6 of ~\cite{sklsh-liu2014sk} for the complete definition of $\operatorname{dist}(K_1, K_2)$, and Equation~4 for $\operatorname{KL}(K_1, K_2)$. 

\begin{lemma}
\label{lemma:dist-vs-cos}
For any two arbitrary random vectors $\vec{p_1}, \vec{p_2} \in R^d$ with angle $\theta_{1,2}$, the probability of  ${\operatorname{dist}(G(\vec{p_1}),G(\vec{p_2}))}$ being less than ${m-l+1 \, (\forall l: 0 \le l \le m)}$ is ${[p(\theta_{1,2})]^l}$.
\end{lemma}
\begin{proof}
According to Definition 5 in ~\cite{sklsh-liu2014sk},
\begin{align*}
\operatorname{dist}(G(\vec{p_1}),G(\vec{p_2})) < m-l+1 
\Leftrightarrow
\operatorname{KL}(G(\vec{p_1}),G(\vec{p_2})) \le m-l
\end{align*}

\noindent Let $\operatorname{KL}(G(\vec{p_1}),G(\vec{p_2})) = m - L$, then we have $m-L \leq m-l$, then $l \leq L$, also trivially $L \leq m$. Thus, by definition of $KL$, $G(\vec{p_1})$ and $G(\vec{p_2})$ share common prefix whose length is $L$.  
This implies that ${h_i(\vec{p_1}) = h_i(\vec{p_2})}$ holds for ${1 \le i \le L}$. 
Due to the fact that each hash function $h_i$ is independently and randomly selected, the desired probability can be computed as follows:
\begin{align}
    \label{eq:dist-and-p_theta}
     & P[\operatorname{dist}(G(\vec{p_1}),G(\vec{p_2})) < m-l+1] = P[\operatorname{KL}(G(\vec{p_1}),G(\vec{p_2})) \leq m-l]\nonumber \\
     &  = \sum_{L=l}^{m} P[\operatorname{KL}(G(\vec{p_1}),G(\vec{p_2})) = m - L] \nonumber\\
     & = \sum_{L=l}^{m-1} \left(\prod_{i=1}^L P[h_i(\vec{p_1}) = h_i(\vec{p_2})]   (1-P[h_{L+1}(\vec{p_1}) = h_{L+1}(\vec{p_2})]) \right) \nonumber \\
     & + \prod_{i=1}^m P[h_i(\vec{p_1}) = h_i(\vec{p_2})] \nonumber \\
     & = \sum_{L=l}^{m-1} p(\theta_{1,2})^L (1-p(\theta_{1,2})) + p(\theta_{1,2})^m = \, p(\theta_{1,2})^l 
\end{align}
\end{proof}

According to Equation~\ref{eq:dist-and-p_theta}, for any given $l$, the probability of two hashkeys $G(\vec{p_1})$ and $G(\vec{p_2})$ being close monotonically increases when the cosine similarity $sim(\vec{p_1}, \vec{p_2})$ increases (i.e., $\theta_{1,2}$ decreases).

Therefore, distance of compound hashkeys is a reasonable metric to estimate cosine similarity between the original vectors. In such context, to find similar data points to the query, searching a small vicinity of the query hashkey is possibly enough. 

\subsection{Extension on hashkey distance}
Though we prove that by using a random projection LSH model as the base model, SK-LSH works for cosine similarity, there is a new problem which does not exist in the original SK-LSH: when using cosine similarity, the hashkey distance defined in \cite{sklsh-liu2014sk} is too coarse to distinguish the actual similarities of many different document embeddings to the query embedding. We name such a problem as ``low resolution problem''. In \cite{sklsh-liu2014sk}, the hashkey distance is defined as 
\begin{align}
\label{eq:hashkey-dist}
    dist(K_1, K_2) = KL(K_1, K_2) + \frac{KD(K_1, K_2)}{C}
\end{align}
where $K_1$, $K_2$ are two hashkeys, $KL(K_1, K_2)$ is the \textit{non-prefix length} \cite{sklsh-liu2014sk} between $K_1$ and $K_2$ (i.e., the length of the sub-sequence after their common prefix), and $KD(K_1, K_2)$ is the $(l + 1)$-th
\textit{element distance} between $K_1$ and $K_2$ (i.e, the absolute difference of the first non-identical elements between the two hashkeys), while $C$ is a constant factor that is set to be larger than the maximum of $KD$ such that $\frac{KD(K_1, K_2)}{C} < 1$ always holds. When the similarity metric is Euclidean distance, each element in a hashkey is an integer within a wide bounded range, which means the range of $KD$ is also wide. But when the similarity metric is cosine similarity, each hashing value must be either 0 or 1, in which case $KD(K_1, K_2) \equiv 1, \forall K_1, K_2$. Therefore, no matter how different $K_1$ and $K_2$ are after the common prefix, $dist(K_1, K_2)$ cannot reflect it.

For example, supposing we use hashkeys of length 6. Given a query embedding $\vec{v_q}$ with hashkey $K_q = 000000$, one document embedding $\vec{v_1}$ with hashkey $K_1 = 111111$ and another $\vec{v_2}$ with $K_2 = 100000$, according to Equation~\ref{eq:hashkey-dist}, we have $dist(K_q, K_1) = dist(K_q, K_2) = 6 + \frac{1}{C}$. In such case, relying on the hashkey distance, $\vec{v_1}$ and $\vec{v_2}$ are the same in terms of similarity to $\vec{v_q}$. But obviously $\vec{v_2}$ is probably more similar to $\vec{v_q}$ than $\vec{v_1}$ to $\vec{v_q}$, since the Hamming distance between $K_2$ and $K_q$ is much smaller than that between $K_1$ and $K_q$. Therefore the original hashkey distance is not able to distinguish between $\vec{v_1}$ and $\vec{v_2}$ in such situations, or in another word, it has a low resolution under the settings of cosine similarity metric and random projection LSH.

To solve this problem, an intuitive way is replacing $KD$ with Hamming distance. But unfortunately, this will cause the linear order of hashkeys to no longer hold. Specifically, the linear order defined by SK-LSH is based on element-wise comparison from the most significant element to the least significant element of the hashkeys. When using random projection LSH (where each element is either 0 or 1), the order is actually a dictionary order (a.k.a., lexicographic order). By sorting hashkeys in this order, SK-LSH guarantees that for any three ordered hashkeys in one array, $K_2$, $K_1$, $K$, $dist(K_2, K) \geq dist(K_1, K)$ always holds, which is essential to the bi-directional expansion search. However, Hamming distance does not consider the element significance, instead, it treats all the elements equally. This may cause that for some ordered $K_2$, $K_1$, $K$, $dist(K_2, K) < dist(K_1, K)$, which breaks the theoretical foundation of SK-LSH. And we do observe many occurrences of this issue in our experiments.

We tackle the low resolution problem successfully by extending the length of sub-sequence used by $KD$. In \cite{sklsh-liu2014sk}, $KD$ is computed as 
\begin{align}
    KD(K_1, K_2) = |k_{1,l+1} - k_{2,l+1}|
\end{align}
where $l$ is the length of the common prefix between hashkeys $K_1$ and $K_2$, and $k_{1,l+1}$, $k_{2,l+1}$ stand for the first non-identical elements between $K_1$ and $K_2$. We extend it to
\begin{align}
    KD_e(K_1, K_2) = |Decimal(K_{1,l+1:l+1+B}) - Decimal(K_{2,l+1:l+1+B})|
\end{align}
where $Decimal(\cdot)$ is the operation to convert a binary hashkey (which only includes 0 and 1 and can be seen as a binary number) into a decimal integer number, while $K_{1,l+1:l+1+B}$ stands for the sub-sequence of length $B$ starting from the $(l+1)$-th element in $K_1$, i.e., right after the common prefix, and the same for $K_{2,l+1:l+1+B}$. Then we have a new definition of hashkey distance
\begin{align}
    dist_e(K_1, K_2) = KL(K_1, K_2) + \frac{KD_e(K_1, K_2)}{C}
\end{align}
here $KL$ keeps original, and we set $C = 2^B$ to make sure $\frac{KD_e(K_1, K_2)}{C} < 1$ still holds. Comparing to the original $KD$ which is always 1, $KD_e$ has more possible values ranging from 0 to $2^B-1$, thus it can better distinguish between different hashkeys. Using $dist_e$, if we set $B = 3$, the example above becomes $dist_e(K_q, K_1) = 6 + \frac{7}{C}$, $dist_e(K_q, K_2) = 6 + \frac{4}{C}$, successfully reflecting that $K_2$ is more similar to $K_q$. 
In addition, $dist_e$ does not change the conclusion from Equation~\ref{eq:dist-and-p_theta}, as $KL$ is the same and $\frac{KD_e}{C}$ is still less than 1. 
Furthermore, unlike Hamming distance, the linear order still holds given $dist_e$. Formally, it can be stated as two lemmas  

\begin{lemma}
\label{lemma:linear-order-holds-on-dist-e-1}
In a hashkey array sorted by the original SK-LSH linear order, for any three ordered hashkeys $K_2$, $K_1$, $K$, $dist_e(K_2, K) \geq dist_e(K_1, K)$.
\end{lemma}

\begin{lemma}
\label{lemma:linear-order-holds-on-dist-e-2}
In a hashkey array sorted by the original SK-LSH linear order, for any three ordered hashkeys $K$, $K_1$, $K_2$, $dist_e(K_2, K) \geq dist_e(K_1, K)$.
\end{lemma}

Lemma~\ref{lemma:linear-order-holds-on-dist-e-1} and \ref{lemma:linear-order-holds-on-dist-e-2} are straightforward to be proven following similar steps to the proof of Lemma 4 and 5 in \cite{sklsh-liu2014sk}, so we just skip the proof in this paper. In conclusion, by extending $dist$ and $KD$, we effectively improve the resolution of the hashkey distance. 

\subsection{Improvement on parallelism}
\label{sec:esk-lsh-improve-parallelism}
In addition to extending the similarity metrics and hashkey distance, ESK-LSH also achieves a higher efficiency than SK-LSH by improving its parallelism. In the original SK-LSH, though it may maintain multiple sorted arrays, the expansion search is not parallel but iterative on the arrays. Specifically, it iteratively looks for the next globally closest hashkey to the query hashkey across all the arrays, which limits its recall within a specific period of search time. In order to efficiently retrieve more candidate hashkeys, we increase the parallelism by making each array independent from others, i.e., ESK-LSH does not look for the globally closest hashkey, but the next closest hashkey locally on each array. Therefore ESK-LSH can do the expansion search parallelly on each array. Such an improvement makes it possible to use more sorted arrays to increase retrieval quality with only tiny time overhead (e.g., the cost of operating system to manage more threads), as long as the hardware resources are sufficient, e.g., enough number of CPU cores. Though in practice the hardware resources might be insufficient such that the parallelism may be limited, we show in the evaluation that a low-end to middle-end machine (with less than 30 cores) is enough for LIDER to outperform the baselines.

\section{RMI and Key re-scaling}
\label{sec:rmi}
As the first learned index for range query, \textit{recursive-model index} (RMI)~\cite{case-for-learned-index-structures-kraska2018case} works on finding out the locations of given keys in some given ordered array, just like what a B-tree index does. But unlike traditional indexes that lookup the locations, RMI ``predicts'' the locations. Specifically, as illustrated in Figure~\ref{fig:arch-core-model}, RMI is a hierarchical structure consists of several layers of simple machine learning models, e.g., shallow neural network or linear regression model. Within one layer, the whole search space (i.e., all locations in the ordered array) is partitioned by all the models in this layer, where each model takes responsibility of a subspace. The key will be first input to the root model. The root model predicts its location, and passes the key to one of the next layer models which corresponds to the subspace containing the location. Then such a predict-and-pass process is repeated layer by layer until one of the final layer models accepts the key and makes the final prediction of the location as the output of RMI. In this process the prediction is gradually refined and finally the error is minimized.
Comparing to traditional indexes like B-tree, RMI has a nearly constant search time with the dataset size. This is because RMI can handle larger dataset by only increasing the width (i.e., the number of models in each layer) while fixing the depth (i.e., the number of layers), and the prediction time mainly depends on the depth.
Specifically, placing more models in one layer will reduce the size of each subspace, such that each model can better fit the smaller subspace. So increasing the width of RMI is often an effective way to improve search quality with only tiny growth on search time.

\subsection{Key re-scaling}
\label{sec:key-rescaling-module}
In order to let RMI accept the hashkeys generated by ESK-LSH, we need to convert each hashkey to a numeric value, so called ``RMI key'' in this paper. We design a \textit{key re-scaling} module for this conversion. In the conversion, the binary hashkey is first seen as a binary number and converted into a decimal integer number, then the decimal integer number is further re-scaled to a floating-point number in the range $[0, L_{array}-1]$, where $L_{array}$ is the length of each sorted hashkey array in ESK-LSH. The re-scaling is the most critical step in the conversion, which has a significant effect on the learning of RMI. This is because on large-scale dataset, the length of hashkey is normally long to have a large enough capacity for encoding the data embeddings. 
Thus after the first-step conversion, the decimal integer numbers are mostly very big. However, comparing to the big integer RMI keys at this step, the labels, i.e., the locations, are much smaller. So most predictions of RMI will be out of range (i.e., much less than 0 or much larger than $L_{array}-1$), making RMI ineffective.

To solve this issue, we use min-max normalization as the second-step re-scaling method. Formally, the normalization is defined as
\begin{align}
    x_{norm} = \frac{x - x_{min}}{x_{max} - x_{min}} (b-a) + a
\end{align}
where $x$ is the original integer key, $x_{min}$ and $x_{max}$ are the minimum and maximum of $x$, $x_{norm}$ is the key after normalization, and $[a, b]$ is the range of $x_{norm}$. In LIDER, we set $a = 0$, $b = L_{array}-1$. The effect of such normalization is evaluated in  Section~\ref{sec:exp-key-rescaling}.

There might be another issue of duplicate RMI keys. Because there is a possibility (which is very low) that LSH generates the same hashkey for different embeddings, duplicate RMI keys might exist. And due to the float-point number precision, after normalization, some originally different RMI keys may also become duplicate. In such cases, duplicate keys still correspond to different locations, i.e., the same training data points have different training labels, which may lower the training quality of RMI. But fortunately, these duplicate keys are adjacent in each sorted array, meaning the error caused by such an issue is bounded in a local range. As long as the local range is small enough, this issue will not affect RMI too much. For instance, we can set the length of hashkey to be large enough to reduce the number of duplicate hashkeys.  

\subsection{Simplified RMI}
The original RMI in \cite{case-for-learned-index-structures-kraska2018case} uses two layers of models, where the top level includes one neural network and the second level includes several linear regression models. We also use a similar two-layer structure, but based on the re-scaled RMI keys, we simplify RMI to only include linear regression models since the RMI keys and their labels/locations are almost linearly distributed, as illustrated in Figure~\ref{fig:rmi-data}. Thus a linear regression on the top can fit better than a neural network, which can also reduce the prediction time. Furthermore, we do not deploy the hybrid index strategy (as introduced in Section~\ref{sec:introduction}) in the RMI for strict error bounding, due to the requirement on high efficiency and the fact that prediction error has been reduced significantly by the key re-scaling module (which is further discussed in Section~\ref{sec:exp-key-rescaling}). 

 \begin{figure}[!h]
  \centering
  \includegraphics[width=0.9\columnwidth]{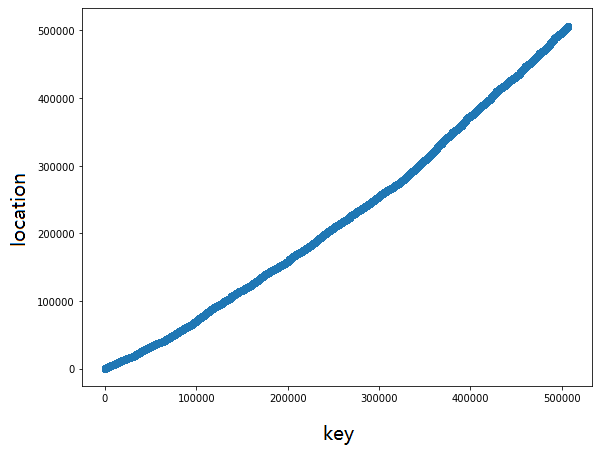}
  \caption{RMI training data distribution (MS-500k dataset as example)}
  \setlength{\belowcaptionskip}{-50pt}
  \label{fig:rmi-data}
\end{figure}

\section{Time complexity}
\label{sec:time-complexity}
We derive the search time complexity of LIDER in this section. Table~\ref{tab:notations} lists major notations used in this and all the following sections. 
The search process of LIDER includes three stages, the centroids retrieval, in-cluster retrieval and the final verification, where the first is simply a standard search process of a single core model, the second is a parallel search with multiple core models (whose time complexity is the same as single core model theoretically), and the third is a simple top-$k$ selection process.    

\begin{table}[!h]
\small
  \begin{tabularx}{\columnwidth}{lX}
    \toprule
     Notation & \multicolumn{1}{c}{Description}  \\
    \midrule
    \ $k_m, k$ & Number of output points by a core model, and that by LIDER  \\
    \midrule
    \ $H$, $M$ & Number of hashkey arrays, and hashkey length in ESK-LSH \\
    \midrule
    \ $R$, $r_0$ & Length of ESK-LSH expansion range on each hashkey array, and a user-specific factor ($R = r_0k_m$)  \\
    \midrule
    \ $c$, $c_0$ & Total number of clusters, and number of retrieved centroids by the centroids retriever  \\
    \midrule
    \ $N$ & Dataset size  \\
    \bottomrule
  \end{tabularx}
  \vspace{1mm}
  \caption{List of notations used in this and following sections}
  \label{tab:notations}
  \vspace{-8mm}
\end{table}

\subsection{Time complexity of a single core model}
The search in a single core model includes five steps: (1) query hashkey generation, (2) query hashkey rescaling, (3) RMI prediciton, (4) ESK-LSH expansion and (5) candidate verification, where the step 2 and 3 are constant with the data scale. So we analyze the time complexity for the rest three steps as follows:
    In \textit{query hashkey generation} step, ESK-LSH generates one query hashkey per array, 
    where each needs $M$ hashings, costing $O(HM)$ time.
    As a linear scanning on range $R$, complexity of \textit{ESK-LSH expansion} step is simply $O(R)$ (since the expansion is parallel on all hashkey arrays).  
    The \textit{candidate verification} computes exact distances between the query and candidates (found by ESK-LSH expansion), and selects top-$k_m$ of them as output. Since there are $RH$ candidates and the output top-$k_m$ are not required to be sorted, this step takes $O(RH)$ time. 
We set $M = \lceil logN \rceil$ to guarantee the hashing space has enough capacity, and $H$ and $r_0$ are fixed. Given that $R = r_0k_m$, we have the overall search time complexity of a single core model
\begin{align}
    O(H\lceil logN \rceil + Hr_0k_m)
\end{align}

\subsection{Time complexity of LIDER search process}
For the centroids retriever, $k_m = c_0$.
And as discussed in Section \ref{sec:exp-cluster-param}, we recommend to set $c_0 = c/100 \sim c/50$, and $c = N/(5\times10^4) \sim N/(10^4)$, i.e., $c_0 = N/(5\times10^6) \sim N/(5\times10^5)$. We present it as $c_0 = t_0N$ for simplification. So the time complexity of centroids retrieval stage is
\begin{align}
    O(H\lceil logN \rceil + Hr_0t_0N)
\end{align}
For each in-cluster retriever we set $k_m = k$,
then time complexity of in-cluster retrieval stage is $O(H\lceil logN \rceil)$
as term $Hr_0k$ is constant with $N$ and this stage is parallelized. 
In addition, an in-cluster retriever can optionally sort the $k$ results, which
is not a burden on time complexity but benefits the next stage. \\
The third stage is simply selecting top-$k$ from the totally $c_0k$ neighbors found by the $c_0$ in-cluster retrievers, with the exact distances computed in the in-cluster retrieval stage. 
As mentioned above, the $c_0$ neighbor lists can be sorted by the in-cluster retrievers, based on which the top-$k$ selection here can be completed in $O(c_0 + klogc_0)$ time, just by utilizing a heap of size $c_0$ that maintains the head of each list.

In summary, the overall search time complexity of LIDER is 
\begin{align}
    & O(2H\lceil logN \rceil + klog(t_0N) + (Hr_0+1)t_0N) 
\end{align}
Given the fact that $t_0$ is tiny ($2 \times 10^{-7} \sim 2 \times 10^{-6}$) and $H$ and $r_0$ are also small (like $H = 10$ and $r_0 < 10$ in our evaluation), the factor $(Hr_0+1)t_0$ is around $10^{-6}$ in practice. This means before $N$ reaches tens $\sim$ hundreds of millions, the complexity is close to logarithmic, then it gets closer to linear (but with a tiny factor) on larger $N$, presenting the high efficiency and scalability of LIDER theoretically. And our evaluation proves this in practice.

\section{Experiments}
\label{sec:exp}
\subsection{Experiment settings}
All experiments are evaluated on a Lambda Quad workstation with 28 3.30GHz Intel Core i9-9940X CPUs, 4 RTX 2080 Ti GPUs and 128 GB RAM. Note that all the experiments are conducted purely on the ANN search stage of dense retrieval, excluding the embedding generating stage, i.e., all of them start from the step where the dense embeddings are already generated. We do not include the deep neural model in the evaluation, but directly use their generated embeddings as input. In addition, all the baseline methods and LIDER are evaluated purely on CPU. We do not utilize GPU for the ANN search in this paper since the linear regressions in LIDER are simple to compute such that it is not necessary to use GPU, though they have the potential to be further accelerated by GPU.  

\subsubsection{Evaluation datasets, tasks and metrics} \hfill\\
\label{sec:eval-ds-and-tasks}
We use two datasets for our evaluation, MS MARCO and Wiki-21M.
\textbf{MS MARCO: }
The MS MARCO passage retrieval (a.k.a., passage ranking) dataset \cite{msmarco-nguyen2016ms} is one of the most commonly used dataset in passage retrieval research. It includes a collection of 8.8M passages from online webpages, several collections of queries used by different tasks and the corresponding groundtruth collections to those queries (i.e., for each query collection, there is a relevant passage collection including relevant (query, passage) pairs). We conduct our experiments on two tasks from MS MARCO, \textit{MS MARCO Dev} and \textit{TREC2019 DL}, which share the same passage collection but use different queries and different performance metrics. Specifically, (1) MS MARCO Dev has 6980 queries where each query has one or more (but only a few) relevant passages, and we measure the performance by MRR@10 for quality and average query processing time (AQT) for efficiency on this task. (2) TREC2019 DL has 43 valid queries and about 9000 (query, passage) pairs in its groundtruth collection. We evaluate the performance using the quality metric NDCG@10 and the efficiency metric AQT on this task.

To explore the impact of data scale on different methods, we sample the MS MARCO dataset into several subsets whose number of passages are respectively 100k, 500k, 1M, 4M (while the full MS MARCO dataset includes 8.8M passages), named as ``MS-'' followed by the size (like ``MS-500k''), where the ``MS-8.8M'' is the full dataset itself. Both of MS MARCO Dev and TREC2019 DL tasks are evaluated on all those subsets. The query and passage embeddings are generated by deep neural model ``msmarco-distilbert-base-v3'' \protect\footnote{available at \protect\url{https://www.sbert.net/docs/pretrained-models/msmarco-v3.html}},  pre-trained on MS MARCO dataset.

\noindent\textbf{Wiki-21M: } 
Another dataset is Wiki-21M from~\cite{DPR-karpukhin-etal-2020}. It includes 21,015,324 passages collected from Wikipedia dump. The evaluation queries we use on this dataset is Natural Questions (NQ)~\cite{DPR-karpukhin-etal-2020} which is designed for end-to-end question answering, with questions collected from real
Google search queries and the answers from Wikipedia articles. The test set of NQ includes 3610 queries. We name this task (i.e., retrieving passages from Wiki-21M dataset to answer the queries in NQ test set) as ``Wiki-21M NQ''. Similar to MS MARCO Dev, the experiment metrics on this task are MRR@10 for quality and average query processing time (AQT) for efficiency. The embeddings of this dataset are generated by pre-trained DPR model~\cite{DPR-karpukhin-etal-2020}  \protect\footnote{the pre-trained model (named as ``checkpoint.retriever.single-adv-hn.nq.bert-base-encoder'') and corresponding Wiki-21M passage embeddings are available at \protect\url{https://github.com/facebookresearch/DPR}}

The embeddings for both MS MARCO and Wiki-21M are 768-dimensional. The total size of MS MARCO passage embeddings is 26GB while that of Wiki-21M embeddings is 62GB.
The embedding similarity metric in all the experiments is cosine similarity. Since our baseline methods do not support cosine similarity but work for inner product similarity, we normalize all the query and passage embeddings such that cosine similarity is equivalent to inner product over the normalized embeddings.

Note that all the time results reported by the experiments are measured after the embedding generation, i.e., they are only the ANN search time, excluding the embedding generation time of the neural model. 

\subsubsection{Baselines}\hfill\\
\label{sec:baselines}
We select several widely used high-dimensional ANN indexes as baselines. They are introduced as follows:
\begin{enumerate}
    \item \textbf{Flat}: It simply searches the exact KNN by brute-force. So we treat its retrieval quality as the upper bound of effectiveness for LIDER and other baseline methods. 
    \item \textbf{PQ}: This is an ANN index that encodes high-dimensional data into shorter codes by product quantization, and does ANN search using the codes to reduce computation.
    \item \textbf{OPQ} \cite{OPQ}: This is an improved variant of PQ index which optimizes PQ to better fit the data by applying a rotation. It has a better effectiveness than PQ. 
    \item \textbf{PCA-PQ} \cite{PCA-PQ}: This is another improved variant of PQ index that applies PCA dimension reduction to the data before encoding it using PQ. 
    \item \textbf{IVFPQ} \cite{ivfadc}: It implements the classic ``inverted index + product quantization'' ANN index, IVFADC in \cite{ivfadc}, which is one of the fastest high-dimensional ANN indexes today.
    \item \textbf{IVFPQ-HNSW}: This method further optimizes the IVFPQ index by using HNSW \cite{Graph-based-index-hnsw-malkov2016} to do the cluster assignment and management for the inverted index, which further improves the search efficiency.
    \item \textbf{FALCONN} \cite{falconn}: This is a mature and high-performance LSH index library based on the classic multi-probe LSH \cite{multiprobe-lsh}, which is one of the most practical LSH methods in real-world applications.
    \item \textbf{SK-LSH}: We include the original SK-LSH index as one of our baselines, in order to evaluate the improvement to it by our idea. 
\end{enumerate}
There are also other popular high-dimensional ANN indexes, but due to some reasons we do not include them. For example, ScaNN \cite{scann} requires users to compile it with AVX support (which provides a non-trivial acceleration using the facility of modern hardware), while LIDER does not utilize such a technique currently since it needs much effort on the engineering. Therefore, to guarantee a fair comparison, we do not include ScaNN in the baselines.

The Flat and all PQ-based baselines are implemented using FAISS \cite{faiss}, a high-performance industrial ANN index library. FALCONN also provides a public codebase, while SK-LSH is implemented by us as there is no open-source implementation. We set the parameters of the baseline indexes as such: (1) Flat is an exhaustively exact search index, so there is no specific parameter to set. (2) For IVFPQ and IVFPQ-HNSW, $C = \sqrt{N}$, $m = 32$, $b = 8$, $p = 500$, where $N$ is the number of passage embeddings in current dataset (as in Table~\ref{tab:notations}), $C$ is the number of centroids associated with the coarse quantizer in IVFADC, $m$ is the number of segments into which each embedding will be split, $b$ is the number of bits used to encode the centroids associated with the product quantizer in IVFADC, and $p$ is the number of nearest inverted file entries to be inspected during search. As recommended by~\cite{faiss}, we dynamically compute $C$ based on dataset size $N$ instead of fixing it. The number of neighbors per node and search depth of the HNSW in IVFPQ-HNSW are both set to be 32. (3) For PQ and OPQ, their parameters are just a subset of those for IVFPQ, i.e., the $m$ and $b$, which are same as IVFPQ. (4) For PCA-PQ, the parameters of its PQ component are the same as OPQ, while its PCA component reduces the data dimension to 192. (5) For FALCONN and SK-LSH, we set the number of hash tables/hashkey arrays $H = 24$ and the hashkey length $M = \lceil log_2(N) \rceil$. Since the memory requirement of SK-LSH on Wiki-21M exceeds the machine limit, we reduce its $H$ to 14 on Wiki-21M.   
For all the experiments in this paper, we set $k=100$, i.e.,always retrieving top-100 relevant passages to each query.

\subsubsection{Experiment categories}\hfill\\
We design two major categories of experiments, 
\begin{enumerate}[leftmargin=2em]
    \item \textbf{end-to-end retrieval evaluation}: This evaluation is conducted on all of MS MARCO Dev, TREC2019 DL and Wiki-21M NQ tasks. It reports retrieval efficiency and quality by the corresponding metrics introduced in Section~\ref{sec:eval-ds-and-tasks} for LIDER and the baselines. To make it complete, this evaluation consists of two parts: (a) evaluation on varying datasets with fixed method parameters, and (b) evaluation on fixed datasets with varying parameters.
    \item \textbf{evaluation of critical parameters, memory usage and construction cost}: This evaluation includes several experiments for the impact of critical parameters on LIDER (Section \ref{sec:exp-esk-lsh}, \ref{sec:exp-key-rescaling} and \ref{sec:exp-cluster-param})
    as well as the memory footprint and construction cost of LIDER (Section \ref{sec:exp-memory-construction}). 
\end{enumerate}

\begin{table*}[t]
\small
\centering
\begin{tabularx}{\textwidth}{l|ccccc|c|ccccc}
\toprule
  & \multicolumn{5}{c|}{\shortstack{\textbf{MS MARCO Dev} \\ \textbf{(MRR@10)}}} & \multicolumn{1}{c|}{\shortstack{\textbf{Wiki-21M NQ} \\ \textbf{(MRR@10)}}} & 
  \multicolumn{5}{c}{\shortstack{\textbf{TREC2019 DL} \\ \textbf{(NDCG@10)}}} \\
\toprule
\textbf{Method} & MS-100k & MS-500k & MS-1M & MS-4M & MS-8.8M & Wiki-21M & MS-100k & MS-500k & MS-1M & MS-4M & MS-8.8M \\ 
\toprule
Flat (\textit{exact}) & 0.8511 &  0.7227 & 0.6496 & 0.4922 & 0.3314
 & 0.5518
 &  0.5681 & 0.4726 & 0.4769 & 0.5762 & 0.6707 \\
\midrule
PQ   & 0.7721 &  0.6304 & 0.5588 & 0.4129 & 0.2734
 & 0.2145
& 0.4585 &  0.4204 &  0.4083 &  0.4533 & 0.5802\\
OPQ   & \textbf{0.8143} &  \textbf{0.6742} & \textbf{0.5994} & \textbf{0.4392} & 0.2907
 & -
& \textbf{0.5658} &  0.4360 &  0.4100 &  0.4893 & 0.5974\\
PCA-PQ   & 0.8001 &  0.6575 & 0.5778 & 0.4248 & 0.2816
 & 0.4513
& 0.5438 &  0.4013 &  0.4080 &  0.4824 & \textbf{0.5997} \\
IVFPQ  & 0.6152 &  0.4811 & 0.4349 & 0.3107 & 0.2154
 & 0.2066
 &  0.4228  & 0.3420  & 0.3584  & 0.4215  & 0.4973 \\
IVFPQ-HNSW  & 0.6151 &  0.4784 & 0.4274 & 0.3138 & 0.212
 & 0.2133
 & 0.3963 &  0.3230 &  0.3433 & 0.3911 &  0.4929\\
FALCONN  & 0.7543 &  0.6402 & 0.5765 & 0.426 & 0.2882
 & 0.3175
 & 0.4712 &  0.3453 &  0.3767 & 0.5232 & 0.5595\\
SK-LSH  & 0.7893 &  0.5785 & 0.5045& 0.3225 & 0.2226
 & 0.2702
 & 0.5386 &  0.3429 &  0.3988 &  0.3806 & 0.4662\\
LIDER  & 0.7428 &  0.6225  & 0.5667 & 0.4292 & \textbf{0.2908}
& \textbf{0.4671}
&  0.4684  & \textbf{0.4548}  & \textbf{0.4366}  & \textbf{0.5363}  & 0.5861\\
\midrule

\bottomrule
\end{tabularx}
\vspace{2mm}
\caption{End-to-end retrieval quality for all three evaluation tasks. The MS MARCO Dev and TREC2019 DL tasks are based on the same subsets of the MS MARCO dataset, while Wiki-21M NQ task is based on the Wiki-21M dataset.}
\vspace{-5mm}
\label{tab:end-to-end-retrieval-quality}
\end{table*}

\subsection{End-to-end retrieval evaluation}
\label{exp:end-to-end-retrieval}
This section introduces the two parts of the end-to-end retrieval evaluation, where the first is on varying datasets and the second is on varying method parameters.
\subsubsection{Evaluation on varying datasets and tasks} \hfill\\
In the first part of the end-to-end retrieval evaluation, the parameters of baseline methods are set as stated in Section~\ref{sec:baselines}, and 
we fix the parameters of LIDER as such (which are selected by grid search): the number of clusters $c = 1000$, the number of retrieved centroids by centroids retriever $c_0 = 20$, the number of ESK-LSH arrays in any core model $H = 10$, and the RMI width (i.e., the number of the second-layer models in an RMI) in centroids retriever $W_c = 10$ while that of each in-cluster retriever $W_i = 5$.
Normally larger $H$ and $W$ lead to better retrieval quality with slightly lower efficiency, while the effect of the clustering related parameters $c$ and $c_0$ is more complex, which will be explored in Section~\ref{sec:exp-cluster-param}.

Table~\ref{tab:end-to-end-retrieval-quality} reports the scores to measure retrieval quality in the three tasks. 
Flat is annotated as \textit{exact} to highlight that it finds the exact k nearest neighbors to the query, therefore we consider its retrieval quality scores as the upper bound for other methods. The other methods are all approximate search methods.
We also highlight the highest score (among the approximate methods only) on each dataset in Table~\ref{tab:end-to-end-retrieval-quality}. OPQ has no score for Wiki-21M NQ, since it requires huge memory space on Wiki-21M which 
exceeds the memory capacity of our evaluation machine. 
According to the scores, LIDER achieves the best retrieval quality among all ANN methods in most cases. Though OPQ is competitive to LIDER in the smaller datasets, LIDER outperforms it on the larger datasets, especially given the fact that OPQ cannot be performed on the largest dataset.   
Therefore, we can still conclude that LIDER has higher effectiveness than all ANN baselines on large-scale data.

\begin{figure}[!t]
  \centering
  \includegraphics[width=\columnwidth]{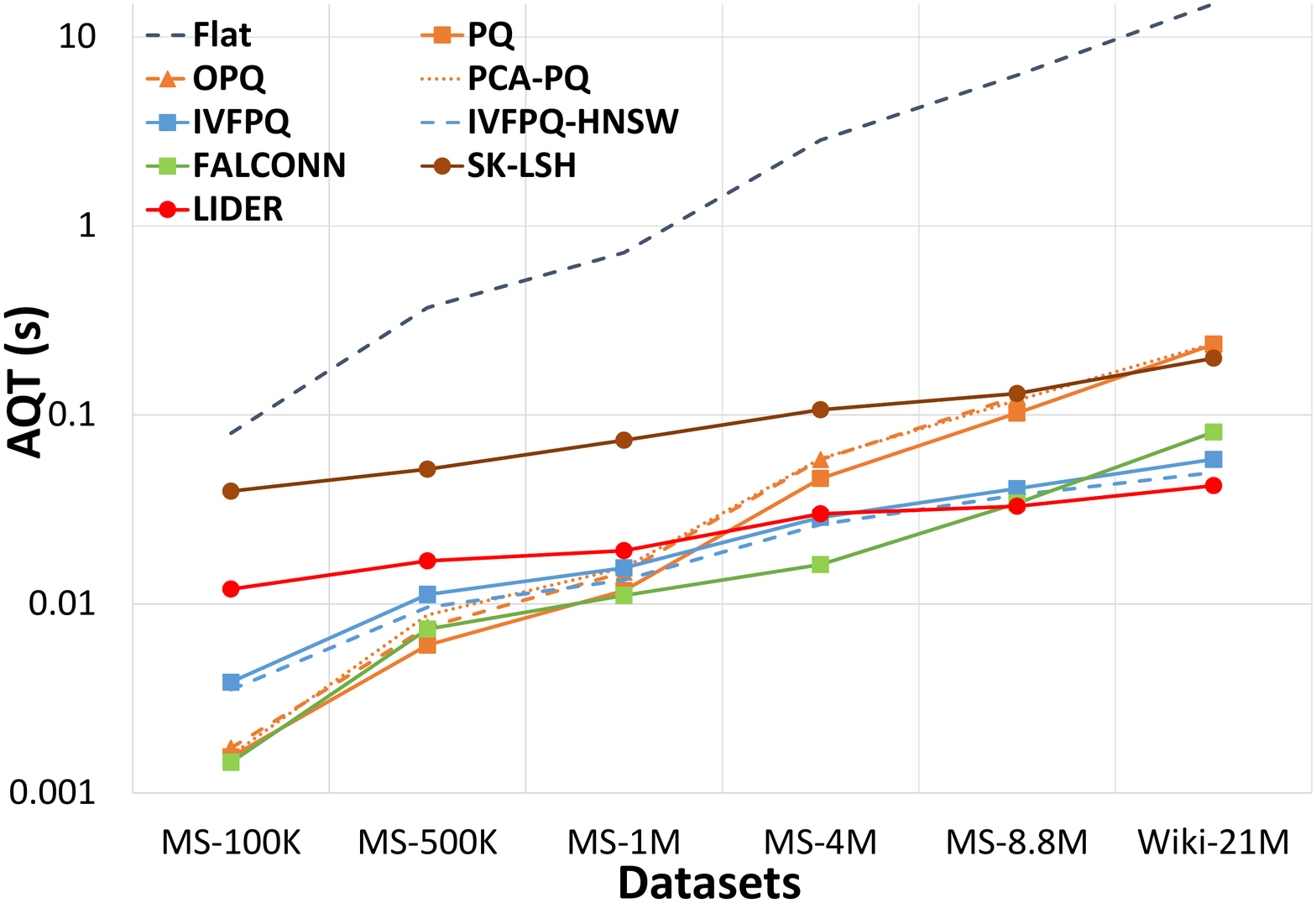}
  \caption{Average query processing time on MS MARCO Dev and Wiki-21M NQ for all the methods}
  \setlength{\belowcaptionskip}{-50pt}
  \label{fig:end-to-end-ms-marco}
\end{figure}

Figure~\ref{fig:end-to-end-ms-marco} illustrates the average query processing time, AQT, for end-to-end retrieval by LIDER and the baselines on MS MARCO Dev and Wiki-21M NQ tasks. 
AQT results of all the methods on TREC2019 DL task are very similar to those on MS MARCO Dev, therefore we do not show them in this paper.

By the figure, Flat (the exact search method, presented by the grey dashed line on the top) takes the longest query processing time, which also increases fastest with the data scale. 
LIDER (the red solid line with circle points) outperforms all baselines on the largest datasets (i.e., MS-8.8M and Wiki-21M) since its average retrieval time grows slowest with data scale, which complies to the time complexity analysis in Section~\ref{sec:time-complexity} and proves the high efficiency of LIDER in practice. Particularly, on the largest datasets (MS-8.8M and Wiki-21M), comparing to the fastest baseline method (i.e., IVFPQ-HNSW which is shown as the blue dashed line), LIDER achieves 15\% $\sim$ 20\% speedup and much higher retrieval quality. Comparing to the highest-quality ANN baseline methods (i.e., OPQ and PCA-PQ which are shown as orange dashed line with triangle points and orange dotted line), LIDER still has better quality on the largest datasets with more significant speedup, i.e., 300\% $\sim$ 500\%. We also observe that SK-LSH has a similar AQT growth trend to LIDER, but its base AQT is too high. Note that in the figure, the distribution of dataset labels on the x-axis are not proportional to the dataset sizes, instead, the labels are evenly placed to have a better view. But this does not affect the comparison of AQT growth trends between different methods. 
Since in real-world applications the data scale can be much larger than that in our experiments, LIDER does have a great potential in highly efficient dense retrieval for real-world scenarios, due to its slow AQT growth trend with the data scale.  

\subsubsection{Evaluation on varying parameters} \hfill\\ 
In the first part experiments above, we fix key parameters of the methods and compare their performance on different datasets and tasks.   
To further evaluate the effectiveness and efficiency of LIDER, in this part we fix the datasets and tasks (i.e., MS MARCO Dev and Wiki-21M NQ), and vary those parameters to draw the AQT-MRR curves for all the approximate methods. The exact method Flat has no parameters to vary its performance, so this part of evaluation does not include it.

Figure~\ref{fig:mrr_vs_time_msmarco_dev} and \ref{fig:mrr_vs_time_wiki_21m} illustrate the AQT-MRR curves for the two tasks respectively. There are some observations: (1) in the range of low MRR, IVFPQ and IVFPQ-HNSW have better efficiency than LIDER, but (2) in the range of high MRR, LIDER is significantly more efficient than the baselines. Actually in our experiments, we find that IVFPQ and IVFPQ-HNSW are hard to achieve a high MRR under the similar resource constraints (e.g., an acceptable length of index building time) to other methods. And that is the reason why the points on the curves of IVFPQ and IVFPQ-HNSW are concentrated in a relatively low MRR area. But it is enough to show the trends, i.e., LIDER has the best performance on effectiveness-efficiency trade-off. In most cases it takes the shortest search time to achieve the same retrieval quality as the baselines, and its efficiency is easy to be further improved with only tiny sacrifice on the quality. Or vice versa, it is also able to significantly enhance the retrieval quality with a tiny loss on search speed. Such a good capability of trade-off does make LIDER practical.        

\begin{figure}[!t]
  \centering
  \includegraphics[width=0.9\columnwidth]{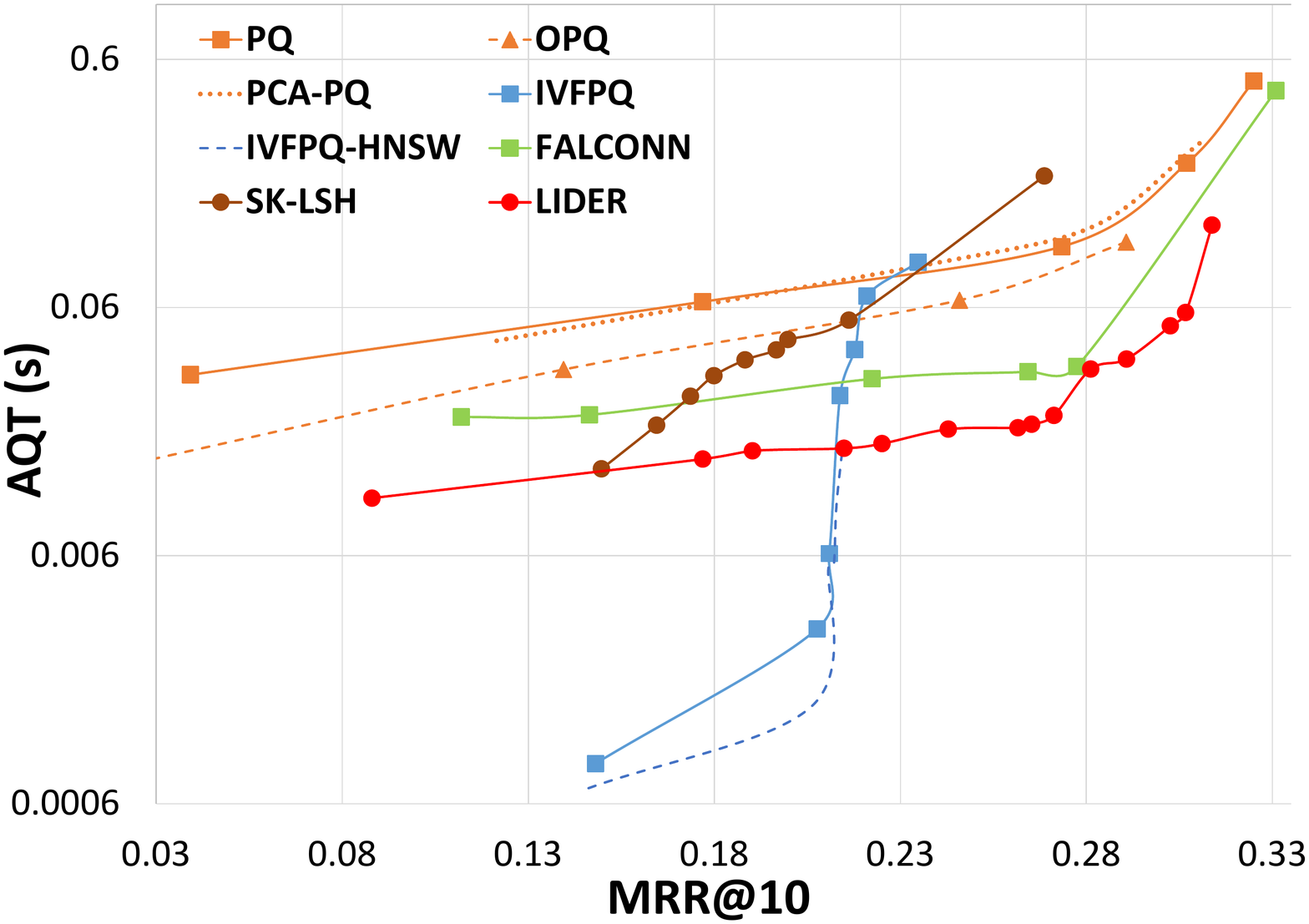}
  \caption{Average query processing time vs MRR@10 for all ANN methods on MS MARCO Dev}
  \setlength{\belowcaptionskip}{-50pt}
  \label{fig:mrr_vs_time_msmarco_dev}
\end{figure}

\begin{figure}[!t]
  \centering
  \includegraphics[width=0.9\columnwidth]{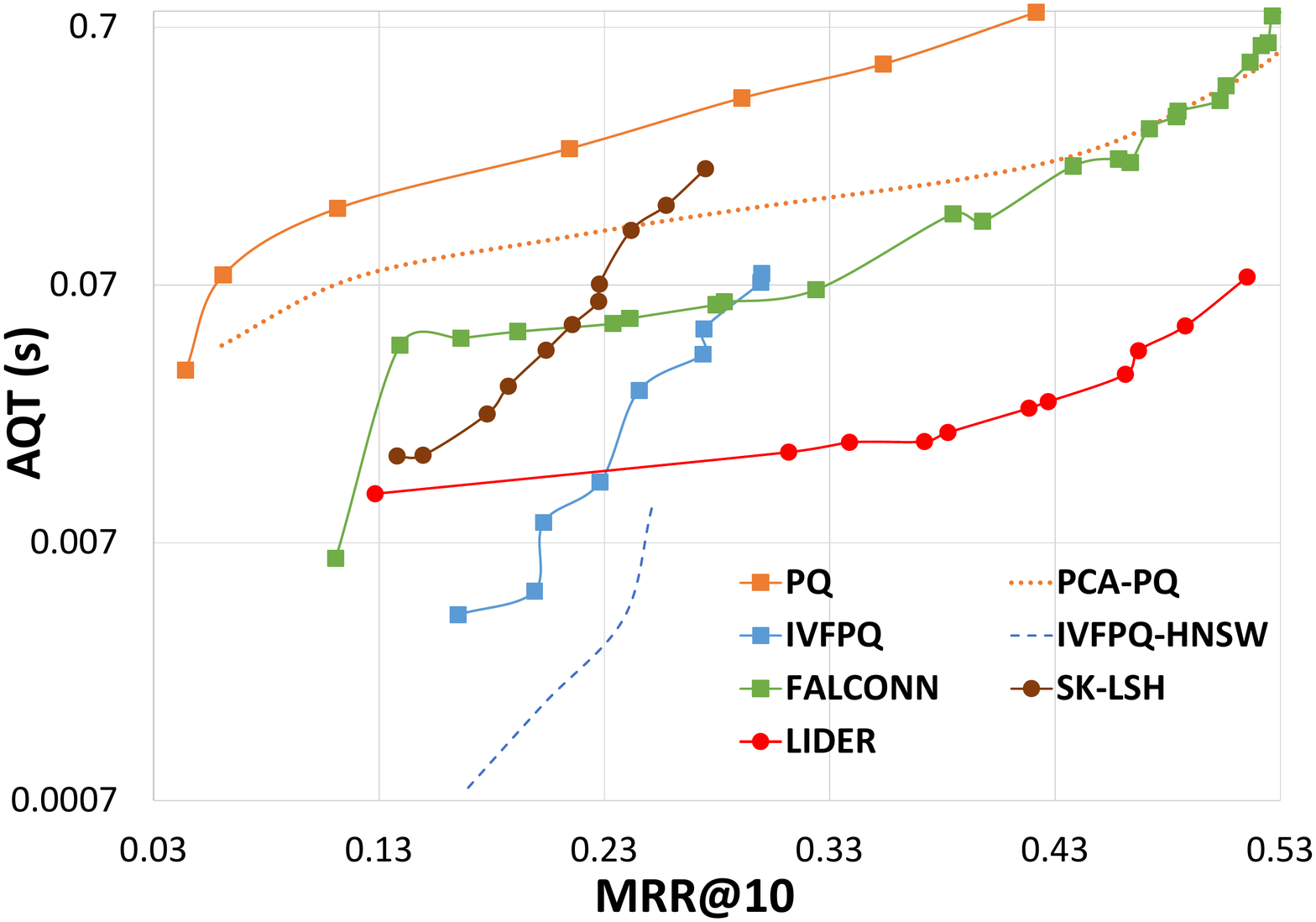}
  \caption{Average query processing time vs MRR@10 for all ANN methods except OPQ on Wiki-21M NQ}
  \setlength{\belowcaptionskip}{-50pt}
  \label{fig:mrr_vs_time_wiki_21m}
    
\end{figure}

\subsection{Impact of $H$ in ESK-LSH}
\label{sec:exp-esk-lsh}

\begin{table}[!h]
\small
  \begin{tabular}{ccc}
    \toprule
     $H$ & MRR@10 & Average expansion time \\
    \midrule
    \ 32 & 0.4928 & 0.0375s  \\
    \ 48 & 0.5569 & 0.0399s  \\
    \ 64 & 0.5912 & 0.0492s  \\
    \bottomrule
  \end{tabular}
  \vspace{1mm}
   \caption{Retrieval MRR@10 and average ESK-LSH expansion time of the standalone core model with different values of $H$  on MS-1M dataset}
   \label{tab:exp-esh-lsh}
  \vspace{-8mm}
\end{table}

The ESK-LSH expansion time takes a large portion in the end-to-end retrieval time of LIDER, so in this experiment, we explore the effect of the critical ESK-LSH parameter, $H$, on the core model performance.
To simplify the evaluation, we build a standalone core model on MS-1M dataset without the clustering-based architecture. So $H$ may be different from the end-to-end evaluation, but the trend of its effect on the performance remains the same. We set $H = 32, 48, 64$ and report MRR@10 and average ESK-LSH expansion time of this single core model in Table~\ref{tab:exp-esh-lsh}, which shows that by using more arrays, the retrieval quality of the core model is significantly improved, with only tiny time overhead added on ESK-LSH expansion. This proves the way to increase the parallelism of SK-LSH (as discussed in Section~\ref{sec:esk-lsh-improve-parallelism}) is effective.  

\subsection{Impact of the key re-scaling module}
\label{sec:exp-key-rescaling}
In this section, we evaluate the effect of key re-scaling module on reducing out-of-range predictions of RMI (as discussed in Section~\ref{sec:key-rescaling-module}).
Specifically, we define a predicted location as \textit{out-of-range prediction} (OOR) if it equals 0 or $L_{array}-1$, as RMI will truncate big prediction to $L_{array}-1$ and round negative prediction to 0. We also define a predicted location as \textit{large-error prediction} (LE) if the gap between it and the true location is larger than 100, i.e., any prediction error exceeding $k$ is a large error where $k = 100$ in our evaluation. Then we check the overlap between OOR and LE predictions, 
which reflects if the large errors are mainly caused by the out-of-range problem. We denote the numbers of OOR, LE and their overlapped predictions by $N_{OOR}$, $N_{LE}$ and $N_{overlap}$. Similar to Section~\ref{sec:exp-esk-lsh}, this experiment is also conducted on a standalone core model.

\begin{table}[!h]
\small
  \begin{tabular}{cccc}
    \toprule
      Using key re-scaling & $N_{OOR}$ & $N_{LE}$ & $N_{overlap}$ \\
    \midrule
    \ No & 4846 & 4733 & 4245  \\
    \ Yes & 3 & 2536 & 0 \\
    \bottomrule
  \end{tabular}
    \vspace{1mm}
    \caption{RMI prediction quality (the number of out-of-range predictions $N_{OOR}$, the number of large-error predictions $N_{LE}$ and the size of their overlap $N_{overlap}$) before and after using key re-scaling module on the MS-100k dataset}
    \label{tab:exp-key-rescaling}
    \vspace{-6mm}
\end{table}

As shown in Table~\ref{tab:exp-key-rescaling}, on MS-100k dataset and the 6980 queries of MS MARCO Dev (where for each query, one prediction is made), before using key re-scaling module, the out-of-range and large-error predictions are heavily overlapped, meaning that most of the large-error predictions are probably caused by the out-of-range problem. And after using the module, all of the three numbers significantly decrease, which means the 
large errors from out-of-range problem have been successfully reduced, and the remaining errors are likely to be from RMI itself. In conclusion, key re-scaling module effectively improves RMI prediction quality.    

\subsection{Impact of the clustering related parameters}
\label{sec:exp-cluster-param}
In this section we investigate the effect of the two clustering related parameters, the number of clusters $c$ and the number of retrieved centroids $c_0$, on the end-to-end performance of LIDER. Experiments are conducted on MS MARCO Dev task using MS-8.8M dataset.

First we fix $c = 1000$ and set $c_0 = 1, 5, 10, 15, 20, 30, 40, 50, 100$ to observe the end-to-end retrieval performance of LIDER. As shown in Figure~\ref{fig:cluster-params-exp-c0}, overall the quality and search time both increase with $c_0$, because increasing $c_0$ will let LIDER retrieve more candidates
on which the exact vector distances are computed and the top-k are finally selected. So more candidates normally lead to more accurate results and longer search time. Another observation in the figure is that the effect of $c_0$ on the retrieval quality improvement is degrading with its growth, i.e., when $c_0$ is small, increasing it will bring more gain on the quality than when it is large, and also bring less latency growth than when it is large. So we recommend users of LIDER to keep $c_0$ relatively small comparing to the total number of clusters $c$, like around $1/100 \sim 1/50$ of $c$.     

\begin{figure}[!t]
  \centering
  \includegraphics[width=0.8\columnwidth]{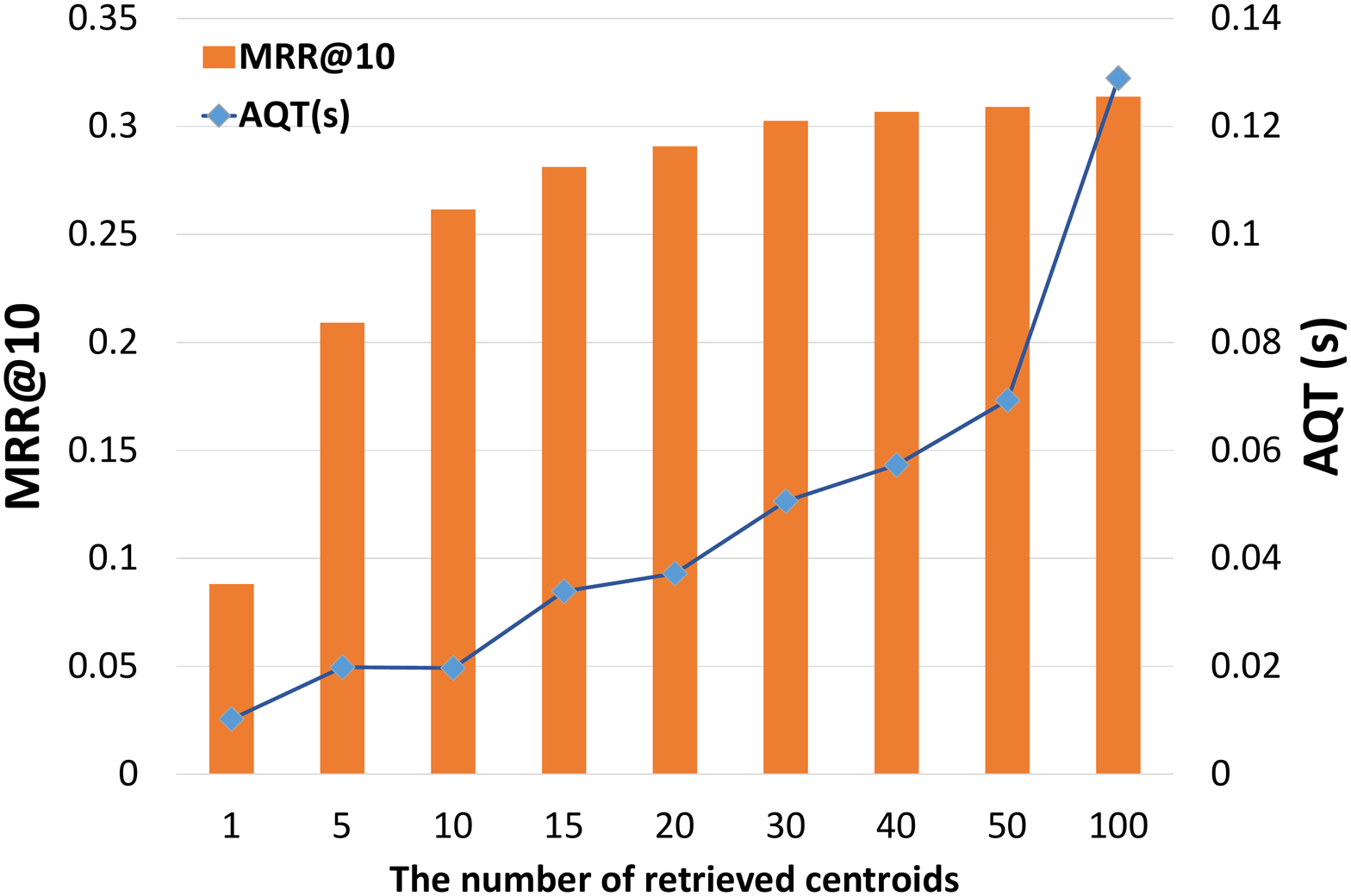}
  \caption{The retrieval quality and average query processing time of different $c_0$ values when $c = 1000$}
  \setlength{\belowcaptionskip}{-50pt}
  \label{fig:cluster-params-exp-c0}
\end{figure}

Second we fix $c_0 = 10$ and vary the total number of clusters $c = 50, 100, 200, 400, 600, 800, 1000, 1200, 1600, 2000, 2400, 3000, 4000$. Figure~\ref{fig:cluster-params-exp-c} presents the MRR@10 and AQT on each $c$ value. The AQT is decreasing with the increasing $c$, which is straightforward to explain: the dataset size is fixed, therefore increasing $c$ means the average size of each cluster will decrease. Given that $c_0$ is fixed, fewer candidates will be retrieved and verified, so AQT decreases. But the MRR@10 is not monotonic, instead, it first increases then falls with the increasing $c$. This is because (1) when $c$ is small, each cluster will include a large number of data points, making it harder for RMI to accurately learn the distribution, 
so the in-cluster retrieval quality degrades. 
And (2) when $c$ is large, each cluster will be small, meaning that the number of candidates is likely not enough since an in-cluster retriever will just stop and return whatever it found when the corresponding cluster is exhausted. In such a case the recall will be probably low, resulting a low MRR. Thus we recommend to set a proper $c$ based on the dataset size, such that each cluster includes around 10k $\sim$ 50k data vectors.  

\begin{figure}[!t]
  \centering
  \includegraphics[width=0.8\columnwidth]{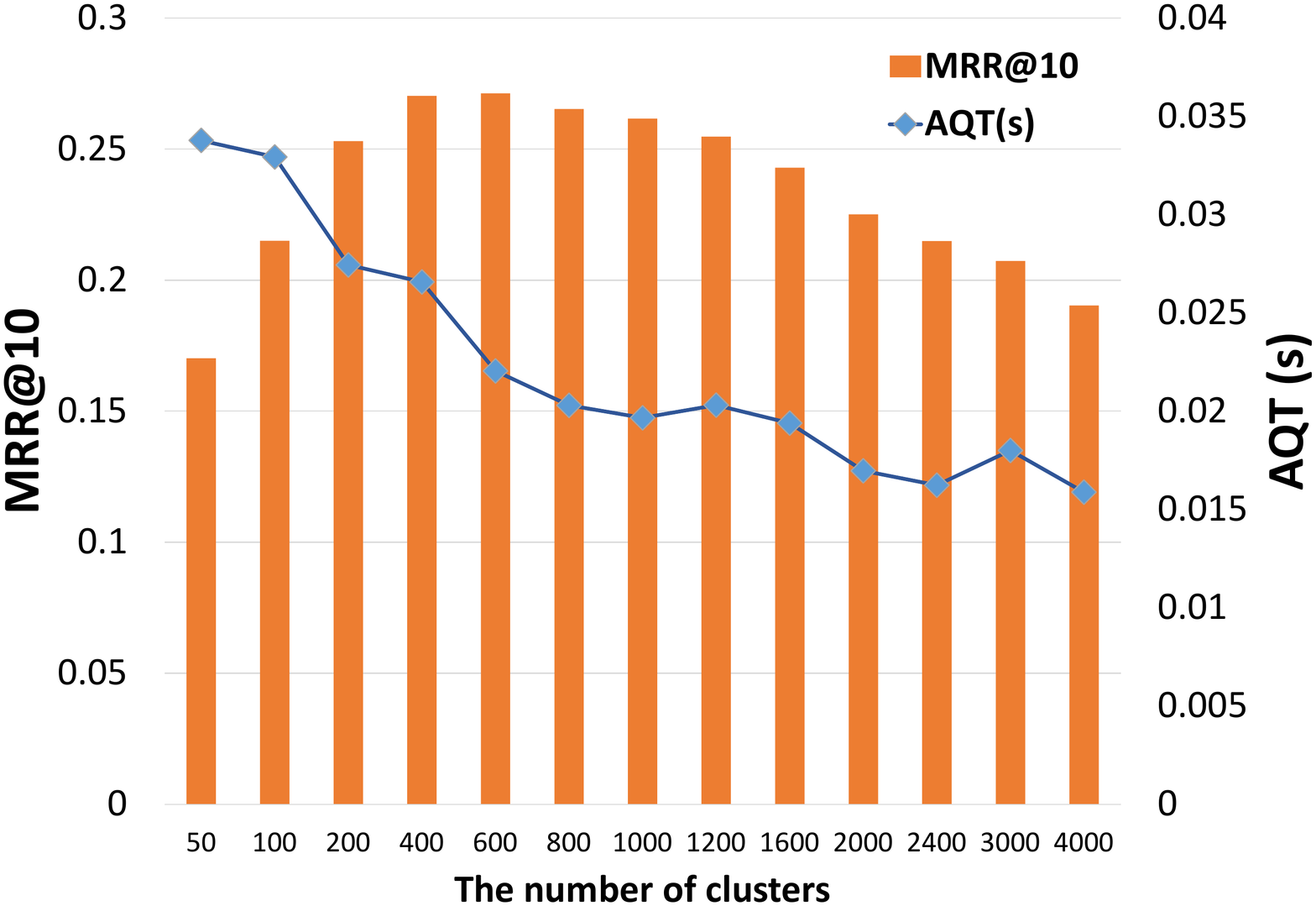}
  \caption{The retrieval quality and average query processing time of different $c$ values when $c_0 = 10$}
  \setlength{\belowcaptionskip}{-50pt}
  \label{fig:cluster-params-exp-c}
\end{figure}

\subsection{Memory footprint and index construction time}
\label{sec:exp-memory-construction}

In this section we report and analyze the memory footprint and index construction time cost of LIDER. Table~\ref{tab:lider-memory} reports elapsed time for the three main construction stages of LIDER (where the construction finishes after Stage 3), as well as memory footprint of LIDER after each stage. In addition, Table~\ref{tab:lider-memory} also includes the construction time and memory usage of the original SK-LSH to show the improvement made by LIDER on the index size. All the memory results in the table have excluded the memory of data embeddings, i.e., they are purely the spaces used to maintain the indexes. The parameters of LIDER and SK-LSH are the same as those in the end-to-end retrieval evaluation (Section~\ref{exp:end-to-end-retrieval}).

\noindent\textbf{Memory footprint:} First, we have the following observations on the LIDER construction stages. By Table~\ref{tab:lider-memory}, memory usage of LIDER after building the centroids retriever (i.e., after Stage 2) is close to the previous stage, meaning that the centroids retriever has a small size, while the memory usage significantly increases after Stage 3, reflecting that the in-cluster retrievers take up the major fraction of the index size. 
Second, comparing to the original SK-LSH, LIDER shows its significance on reducing memory overhead. Specifically, it saves 53\% and 58\% memory space comparing with SK-LSH on the two largest datasets, enabling it to process larger scale data than original SK-LSH on the same hardware. Such a saving is achieved mainly by its cluster-based architecture. As each cluster is significantly smaller than the whole dataset, the required number of hashkey arrays and the length of hashkey can be both reduced while still guaranteeing high effectiveness, e.g.,  in Section~\ref{exp:end-to-end-retrieval} each core model of LIDER only needs 10 arrays to outperform the SK-LSH baseline with 24 arrays.

\noindent\textbf{Construction time:} LIDER presents its significant impact on saving memory usage, with a cost of construction time. 
As in Table~\ref{tab:lider-memory}, we have observed that LIDER has a longer construction time 
than SK-LSH, where the bottleneck is the k-means clustering (Stage 1). However, this is not an issue in practice since there are many approaches to speed up the k-means clustering. For example, the FAISS library provides GPU-based k-means clustering, which can complete LIDER Stage 1 in a few seconds.

\begin{table}[h]
\small
\centering
\begin{tabularx}{\columnwidth}{X|cc|cc}
\toprule
& \multicolumn{2}{c|}{\textbf{MS-8.8M}} & \multicolumn{2}{c}{\textbf{Wiki-21M}} \\
\toprule
 & Time & Memory & Time & Memory \\
\midrule
LIDER Stage 1 - Clustering & 818s & 306MB & 1336s & 725MB\\
LIDER Stage 2 - Building CR & 0.1s & 309MB & 0.1s & 727MB\\
LIDER Stage 3 - Building all IRs & 83s & 14.2GB & 224s & 26.4GB\\
\midrule
SK-LSH & 472s & 30.2GB & 1181s & 62.7GB\\
\bottomrule
\end{tabularx}
\vspace{1mm}
\caption{Elapsed time of each construction stage in LIDER and memory usage by LIDER when each stage finishes, with comparison to the original SK-LSH on the two largest datasets. The \textit{CR} in LIDER Stage 2 means \textit{Centroids Retriever} while the \textit{IR} in Stage 3 stands for \textit{In-cluster Retriever}.}
\label{tab:lider-memory}
\vspace{-8mm}
\end{table}

\section{Conclusion}
In this paper, we introduce LIDER, an efficient high-dimensional learned index for large-scale dense passage retrieval. 
Experiments present that LIDER outperforms the state-of-the-art ANN indexes on large-scale dense retrieval tasks by achieving higher efficiency with high retrieval quality. It also has a much better capability of effectiveness-efficiency trade-off. Therefore, LIDER can serve as a practical and powerful component in very large-scale real-world dense retrieval applications. 

\begin{acks}
This work is partially supported by DARPA under Award \#FA8750-18-2-0014 (AIDA/GAIA). 
\end{acks}

\bibliographystyle{ACM-Reference-Format}
\bibliography{main}

\end{document}